\documentclass[11pt]{article}
\usepackage[utf8]{inputenc}

\usepackage{amsmath, amssymb, amsthm, amsfonts}
\usepackage{mathtools}
\usepackage{thm-restate}

\usepackage[dvipsnames,svgnames,x11names]{xcolor}
\definecolor{ForestGreen}{rgb}{0.1333,0.5451,0.1333}
\usepackage[colorlinks=true,linkcolor=ForestGreen,citecolor=blue]{hyperref}

\usepackage[margin=1in]{geometry}
\usepackage{graphics}
\usepackage{pifont}
\usepackage{tikz}
\usepackage{bbm}
\usepackage[T1]{fontenc}

\usetikzlibrary{arrows.meta}
\usepackage{environ}
\usepackage{framed}
\usepackage{url}
\usepackage{algorithm}
\usepackage[algo2e,linesnumbered,lined,ruled,boxed]{algorithm2e}
\usepackage{algpseudocode}
\usepackage[noabbrev,capitalize,nameinlink]{cleveref}
\crefname{equation}{}{}
\usepackage[labelfont=bf]{caption}
\usepackage{cite}
\usepackage{framed}
\usepackage[framemethod=tikz]{mdframed}
\usepackage{appendix}
\usepackage{graphicx}
\usepackage[textsize=tiny]{todonotes}
\usepackage{tcolorbox}
\allowdisplaybreaks[1]

\newcommand\remove[1]{}

\newtheorem{lemma}{Lemma}[section]
\newtheorem{theorem}{Theorem}
\newtheorem*{lemma*}{Lemma}
\newtheorem{corollary}[lemma]{Corollary}
\newtheorem*{corollary*}{Corollary}

\newtheorem{conj}[lemma]{Conjecture}
\newtheorem*{remark}{Remark}

\theoremstyle{definition}

\newtheorem*{theorem*}{Theorem}
\newtheorem{definition}[lemma]{Definition}

\newtheorem*{rem*}{Remark}

\newtheorem{obs}{Observation}

\newcommand{\eps}{\varepsilon}
\newcommand{\R}{\mathbb{R}}

\newcommand{\E}{\mathop{\mathbb{E}}}

\crefname{algocf}{Algorithm}{Algorithms}
\newcommand{\supp}{\mathsf{supp}}
\newcommand{\Z}{\mathbb{Z}}

\newcommand{\bbC}{\mathbb{C}}

\newcommand{\A}{\Sigma}
\newcommand{\B}{\Gamma}
\newcommand{\C}{\Phi}
\renewcommand{\a}{\sigma}

\renewcommand{\bar}{\overline}
\renewcommand{\hat}{\widehat}

\newcommand{\wt}{\widetilde}

\renewcommand{\bar}{\overline}

\newcommand{\stab}{\mathsf{Stab}}
\renewcommand{\ge}{\geqslant}
\renewcommand{\le}{\leqslant}
\renewcommand{\geq}{\geqslant}
\renewcommand{\leq}{\leqslant}

\newcommand{\opt}{\textsc{Opt}}

\newcommand\Expect[2]{{\mathop{\mathbb{E}}_{#1}\left[ {#2} \right]}}

\newcommand\p{\mbox{\bf P}\xspace}
\newcommand\np{\mbox{\bf NP}\xspace}

\newtheorem{question}{Question}

\begin{document}

\title{On Approximability of Satisfiable $k$-CSPs: VII}

\author{Amey Bhangale\thanks{Department of Computer Science and Engineering, University of California, Riverside. Supported by the Hellman Fellowship award.}
	\and
	Subhash Khot\thanks{Department of Computer Science, Courant Institute of Mathematical Sciences, New York University. Supported by
		the NSF Award CCF-1422159, NSF CCF award 2130816, and the Simons Investigator Award.}
	\and
    Yang P. Liu\thanks{School of Mathematics, Institute for Advanced Study, Princeton, NJ. This material is based upon work supported by the National Science Foundation 
under Grant No. DMS-1926686}
    \and 
	Dor Minzer\thanks{Department of Mathematics, Massachusetts Institute of Technology. Supported by NSF CCF award 2227876 and NSF CAREER award 2239160.}}
\date{\vspace{-5ex}}
\clearpage\maketitle

\begin{abstract}
Let $\Sigma_1,\ldots,\Sigma_k$ be finite alphabets, and let $\mu$ be a distribution over $\A_1 \times \dots \times \A_k$ in which the probability of each atom is at least $\alpha$. We prove that if $\mu$ does not admit Abelian embeddings, and $f_i: \A_i \to \bbC$ 
are $1$-bounded functions (for $i=1,\ldots,k$) such that 
\[ \left|\E_{(x_1,\dots,x_k) \sim \mu^{\otimes n}}\Big[f_1(x_1) \dots f_k(x_k)\Big]\right| 
\geq \eps, \]
then there exists $L\colon \A_1^n\to\mathbb{C}$ of degree at most $d$ and $\|L\|_2\leq 1$ such that 
$|\langle f_1, L\rangle|\geq \delta$, where $d$ and $\delta>0$ depend only on 
$k, \alpha$ and $\eps$. This answers the analytic question posed by Bhangale, Khot, and Minzer (STOC 2022). We also prove several extensions of this
result that are useful in subsequent applications.
\end{abstract}

\pagenumbering{gobble}

\newpage

\setcounter{tocdepth}{2}

\normalsize
\pagebreak
\pagenumbering{arabic}

\section{Introduction}\label{sec:intro}
This paper continues the investigation of the approximability of constraints satisfaction problems~\cite{BKM1, BKM2, BKM3, BKM4, BKM5, csp6}. 
While previous papers dealt with $3$-ary predicates, the focus of the current paper is 
on $k$-ary predicates for $k>3$. Our primary contribution is a set of new analytical 
inequalities for a general family of $k$-ary distributions, extending results of Mossel~\cite{Mossel10} about the class of connected distributions.

\subsection{Constraint Satisfaction Problems}
\subsubsection*{The Decision Version}
Constraint satisfaction problems (CSPs in short) are some of the most fundamental problems in computer science. For a finite alphabet $\Sigma$ and a predicate $P : \Sigma^k \rightarrow \{0,1\}$, an instance $\Psi$ of the problem $P$-CSP consists of a set of variables $x_1, x_2, \ldots, x_n$ and a collection of constraints $C_1, C_2, \ldots, C_m$, where each constraint is of the form $P(x_{i_1}, x_{i_2}, \ldots, x_{i_k}) = 1$. The goal is to 
decide, given an instance $\Psi$, whether there is an assignment $A\colon \{x_1,\ldots,x_n\}\to\Sigma$ satisfying all of the constraints or not. For a collection of predicates $\mathcal{P}\subseteq \{P\colon \Sigma^k\to \{0,1\}\}$, the problem $\mathcal{P}$-CSP is defined analogously, where now
each constraint $C_i$ takes the form $P(x_{i_1}, x_{i_2}, \ldots, x_{i_k}) = 1$ for some $P\in\mathcal{P}$.\footnote{Some well known problems, such as $3$-SAT for example, involve taking negations of variables. To put such problems in the framework we discuss, one takes a collection $\mathcal{P}$ of $8$ predicates, corresponding to the predicate $P(x,y,z) = x\lor y\lor z$ applied on all negation patterns on $3$ variables.}

A systematic study of the complexity of solving CSPs began by Schaefer~\cite{Schaefer78}, who demonstrated that for any $\mathcal{P}$ with $|\Sigma|=2$, the problem  $\mathcal{P}$-CSP is either in $\p$ or else is $\np$-complete. The well-known Dichotomy Conjecture by Feder and Vardi~\cite{FederV98}, suggested that this assertion holds for all finite alphabets $\Sigma$. 
This conjecture was recently confirmed by Bulatov and independently by Zhuk~\cite{Bulatov17,Zhuk20},  following a long line of work using the methods of abstract algebra.

\subsubsection*{The Optimization Version}
In the maximization version of $\mathcal{P}$-CSP, called Max-$\mathcal{P}$-CSP, the task is to find an assignment to the variables that satisfies the maximum fraction of the constraints. An $\alpha$-approximation algorithm is a polynomial-time algorithm which always returns an assignment satisfying at least $\alpha\cdot \opt$ fraction of the constraints, where $\opt$ is the value of the optimum assignment. The focus of the current work is on approximability of fully satisfiable instances, meaning the case that the algorithm is guaranteed that $\opt = 1$.

The PCP Theorem \cite{FGLSS96, AroraLMSS1998, AroraS1998} implies that it is \np-hard to approximate many CSPs within some constant factor $\alpha<1$. H{\aa}stad~\cite{Hastad2001} showed that for specific CSPs, one can obtain hardness results for much better approximation ratios. For instance, he showed that for any $\eps>0$, 3-SAT cannot be approximated within factor $\frac{7}{8}+\eps$ in polynomial time, unless $\p=\np$.\footnote{Note that a random assignment satisfies $\frac{7}{8}$-fraction of the clauses in expectation, so H{\aa}stad's result in fact asserts that the best approximation algorithm for $3$-SAT is (essentially) obtained by randomly sampling an assignment.} His result~\cite{Hastad2001} is in fact stronger: even if the formula is promised to be {\em fully satisfiable}, it is still \np-hard to find an assignment satisfying more than $(\frac{7}{8}+\eps)$-fraction of the clauses. In another notable result, H{\aa}stad showed given 3-LIN instance\footnote{This CSP is over $\mathbb{F}_2$, with constraints of the type $x_{i_1} \oplus x_{i_2} \oplus x_{i_3} = 1$ or $0$.} promised to be $(1-\eps)$-satisfiable, it is \np-hard to find an assignment satisfying more than $(\frac{1}{2}+\eps)$-fraction of the constraints. Unlike 3-SAT, under the stronger promise that the instance is fully satisfiable, one can efficiently find an assignment that satisfies all constraints. Thus, knowing that a CSP instance is fully satisfiable sometimes allows for better approximation algorithms.

By now, the theory of approximability of \emph{almost satisfiable} CSPs is pretty well understood. 
Indeed, assuming the Unique-Games Conjecture~\cite{Khot02}, we know that the best possible polynomial time algorithms for any CSP 
(in the almost satisfiable regime) is given by a semi-definite programming algorithm (SDP in short). The first such result is from~\cite{KKMO07,MOO05}, where it is proved that assuming the Unique-Games Conjecture, the Goemans and Williamson~\cite{GW95} approximation algorithm for the Max-Cut problem is essentially the best possible. Raghavendra~\cite{Rag08} generalized this result to {\em all} CSPs. To do so, he formulated a
natural SDP relaxation of any given CSP instance, and showed that a generic randomized rounding procedure achieves the best possible approximation ratio. Towards this end, he showed that any $(c, s)$ integrality gap instance for this relaxation can be transformed to a hardness result: assuming the Unique-Games Conjecture, for all $\eps>0$, given an instance promised to be $(c-\eps)$-satisfiable, it \np-hard to find an assignment that satisfies at least $(s+\eps)$-fraction of the constraints. 

In comparison, the theory of approximability of satisfiable CSPs is 
much less developed. Besides results
in a few special cases, such as the result of  H{\aa}stad~\cite{Hastad2001} discussed above, not much is known, even assuming conjectures in the spirit
of the Unique-Games Conjecture. One of the primary goals of the recent line of research~\cite{BKM1, BKM2, BKM3, BKM4, BKM5, csp6} is to further develop this
theory for general classes of predicates. 
Ultimately, the
hope is to establish analogous results to Raghavenra's Theorem~\cite{Rag08} (but for the case of satisfiable instances) under conjectures such as 
the Rich $2$-to-$1$ Games Conjecture~\cite{BravermanKM21} (see also~\cite{BravermanKLM21}). The current 
paper is a continuation of this line of research.



\subsection{Towards Satisfiable CSPs}

The most important building-block in Raghavendra's result (as well as in many other PCPs) is the construction of a suitable dictatorship test. A function $f: \Sigma^n \rightarrow \Sigma$ is called a dictatorship function if $f(x)$ depends only on one of $x$'s coordinates. A dictatorship test is a randomized procedure which queries $f$ at a few (correlated) locations randomly and, based on these, decides if $f$ is a dictatorship function or far from any dictator function.

There are three important properties of the test which are useful in getting hardness of approximation results for CSPs. The first one is the {\em completeness parameter} $c$, which is the probability that the test accepts a dictatorship function. The second property is the {\em soundness parameter} $s$, which is the maximum probability the test accepts a far-from-dictatorship function. The third property is the decision predicate that the test uses. If the test uses predicates from the class $\mathcal{P}$,  has completeness $c$ and soundness $s$, then it can be used to prove a UG-hardness result for Max-$\mathcal{P}$-CSP with completeness $(c-\eps)$ and soundness $(s+\eps)$, for any constant $\eps>0$. 

Raghavendra proved his result by showing how to transform a $(c, s)$ integrality gap instance for Max-$\mathcal{P}$-CSP to a dictatorship test with completeness $(c-\eps)$, soundness $(s+\eps)$ and predicates $\mathcal{P}$, and then combined it with a reduction from Unique-Games. His technique inherently cannot be used to prove hardness results 
for satisfiable instances. For once, using the Unique-Games Conjecture, one inherently loses perfect completeness, as Unique-Games are easy on satisfiable instances. However, this may not be a very fundamental loss, as there is a feasible variant of the Unique-Games Conjecture, known as the Rich $2$-to-$1$ Games Conjecture~\cite{BravermanKM21}, that has perfect completeness. It is plausible that assuming the Rich $2$-to-$1$ Games Conjecture, one 
could transform dictatorship tests into hardness results without any loss in the completeness. Therefore, it seems that the more fundamental place 
where perfect completeness is lost is in the design of the dictatorship test.

The paper~\cite{BKM1} initiated a systematic study of completely characterizing the precise approximability of every $k$-ary CSP on satisfiable instances. 
They considered natural dictatorship tests for a large class of predicates similar to the ones given by Raghavendra, except that they do not lose the perfect completeness which we now describe. The starting point is an instance $\Psi$ of $P$-CSP, whose SDP value is $1$, and the integral value (i.e., maximum fraction of the constraints that can be satisfied by an assignment) is $s$. The SDP solution consists of vectors as well as local distribution for each constraint. Since the SDP value is $1$, all these local distributions are supported on the satisfying assignments to $P$. Let $\mu_i$ be the local distribution corresponding to the $i^{th}$ constraint of the instance. The test is as follows (here $\eps>0$ is a small constant independent of $n$):\\

\fbox{
	
	\parbox{430pt}{
		Given $f:\Sigma^n \rightarrow \Sigma$,
		\begin{enumerate}
			\item Select a constraint $C_i$ from $\Psi$ with probability proportional to its weight.
			\item Construct a $k\times n$ matrix by sampling each column of the matrix independently according to $\mu_i$.
            \item Let $x_j$ be the $j^{th}$ row of the matrix.
			\item Check if $P(f(x_1), f(x_2), \ldots, f(x_k)) = 1$.
		\end{enumerate}
		
	}
}
\\

If $f$ is a dictatorship function, then the test accepts with probability $1$. This follows because for every $i$, the distribution $\mu_i$ is supported on the satisfying assignments to $P$ and therefore every column of the matrix is from $P^{-1}(1)$. What about the acceptance probability of the test when $f$ is far from dictator functions? 
When a small loss in the completeness parameter is allowed, one can modify the distributions $\mu_i$ slightly and ensure that they are fully supported on $\Sigma^k$, in which case one can appeal to results of Mossel~\cite{Mossel10}. When no loss in completeness is allowed, one has to work with the distribution $\mu_i$ as is, and this task becomes significantly harder.

The paper~\cite{BKM1} relates the soundness analysis of the test to the following $k$-wise correlation of functions $f_i : \A_i^n \rightarrow [-1,1]$ with respect to the distribution $\mu$ supported on $P^{-1}(1)$. Specifically, the following general question was posed in~\cite{BKM1}.

\begin{question}\label{question:1} {\rm (Informal)}  Find the necessary and sufficient condition on a distribution $\mu$ over $\A_1 \times \dots \times \A_k$, such that
	\begin{equation}
		\label{eq:exp_prod_hd}
		\left| \Expect{(x_1, x_2, \ldots, x_k)\sim \mu^{\otimes n}}{f_1(x_1)f_2(x_2)\cdots f_k(x_k)} \right| \to 0  \ \ \ \ \ \mbox{as} \ \ \ \ d \to \infty,
	\end{equation}
	for all functions $f_i : \A_i^n \rightarrow [-1,1]$ such that at least one function (essentially) has degree at least $d$.
\end{question}
Mossel~\cite{Mossel10} showed that if the distribution $\mu$ is {\em connected}, then~\eqref{eq:exp_prod_hd} as
above holds. Here, a distribution $\mu$ is called connectedness if looking at the graph whose vertices consist of $\supp(\mu)$, and two vertices are adjacent if they differ on exactly one coordinate, is connected. The connectedness condition, however, is not necessary, as noted implicitly in~\cite{BhangaleK21}. Let $G$ be a non-Abelian group
with no dimension one representation, and take $P: G^3 \to \{0,1\}$ as
$P^{-1}(1) = \{ (x,y,z)\ |\ x\cdot y \cdot z = 1_G\}$, along with the distribution $\mu$ that is uniform on $P^{-1}(1)$. The distribution $\mu$ is not connected, but yet~\cite{BhangaleK21} prove that~\eqref{eq:exp_prod_hd} still holds.

\subsection{Abelian Embeddings}
Motivated by~\cref{question:1}, the work~\cite{BKM1} suggested the notion of 
\emph{Abelian embeddings} of distributions, defined as follows.

\begin{definition}\label{def:embed}
 We say that a distribution $\mu$ on $\A_1 \times \dots \times \A_k$
  admits an Abelian embedding if there exist an Abelian group $G$ and mappings $\sigma_i : \A_i \rightarrow G$, $1 \leq i \leq k$ not all constant such that for all $(a_1, a_2, \ldots, a_k)\in \mathsf{supp}(\mu)$, it holds that $\sum_{i=1}^k \sigma_i(a_i) = 0_G$.
\end{definition}
First, we note that for~\eqref{eq:exp_prod_hd} to hold, it is necessary for $\mu$ to not admit any Abelian embedding. Indeed, suppose that $\mu$ admits an Abelian embedding, and fix $G$ and $\sigma_1,\ldots,\sigma_k$ as in~\cref{def:embed}. Fix some non-trivial character $\chi\in \hat{G}$ and define $f_i(x_i) = \prod_{j=1}^n\chi(\sigma_i((x_i)_j))$. 
Note that for each $(x_1, x_2, \ldots, x_k)\in \supp(\mu^{\otimes n})$, the product $f_1(x_1) f_2(x_2)\cdots f_k(x_k)$ is equal to
	\begin{align*}
		\prod_{i=1}^k \prod_{j=1}^n\chi(\sigma_i((x_i)_j)) 
		= \prod_{j=1}^n \prod_{i=1}^k\chi(\sigma_i((x_i)_j)) 
		=  \prod_{j=1}^n\chi\left(\sum_{i=1}^k\sigma_i((x_i)_j)\right) 
        = \prod_{j=1}^n \chi(0_G) 
        =  1.
	\end{align*}
In particular, the expectation of $f_1(x_1) f_2(x_2)\cdots f_k(x_k)$ over $\mu^{\otimes n}$ is 
bounded away from $0$ in absolute value. Yet, for large $n$, for each $i$ such that $\sigma_i$ is not constant the corresponding $f_i$ is a (essentially)
 high-degree function.\footnote{The functions here are complex-valued with absolute value $1$; one can appropriately take their real or imaginary parts if one insists on
 having real-valued functions.}
Motivated by these examples and a long-term application to approximability of CSPs on satisfiable instances, the authors of \cite{BKM1} hypothesized
that having no Abelian embeddings is also a sufficient condition. More precisely:
\begin{conj}\label{conj:k-ary-hyp} {\rm (Informal):}
For a distribution $\mu$ on $\Sigma^k$,  Conclusion \eqref{eq:exp_prod_hd} holds if and only if $\mu$ admits no Abelian embedding.
\end{conj}
~\cref{conj:k-ary-hyp} was proved for $k=3$ in~\cite{BKM2}. Follow-up works~\cite{BKM3, BKM4, BKM5, csp6} further studied the case of $k=3$ and distributions $\mu$ that do admit Abelian embeddings, and gave a characterization of functions $f_1,f_2,f_3$ that may exhibit such a correlation 
(provided that the distribution $\mu$ is pairwise-connected, which is a much milder requirement).

 \subsection{Our Contributions}
 \subsubsection{Main Result}
Our first result is a confirmation of~\cref{conj:k-ary-hyp} for all $k$.
 \begin{theorem} [Main theorem]
\label{thm:kcsp}
Let $k$ be a positive integer and let $\mu$ be a distribution over $\A_1 \times \dots \times \A_k$ that does not admit any Abelian embedding, and in which the probability of each atom is at least $\alpha$. Then, for every $\eps > 0$ there is $\delta := \delta(\alpha,\eps)>0$ such that for sufficiently large $n$, if $1$-bounded functions $f_i: \A_i^n \to \bbC$ satisfy
\[ \left|\E_{(x_1,\dots,x_k) \sim \mu^{\otimes n}}\Big[\prod_{i=1}^k f_i(x_i)\Big]\right| \ge \eps, \]
then $\stab_{1-\delta}(f_i) \ge \delta$ for all $i = 1, 2, \dots, k$. Quantitatively, we can take 
\[
\delta := \exp(-\exp(\dots\exp(\eps^{-O_{\alpha}(1)})))
\]
where the number of exponentials is at most $k^{O(k)}$.
\end{theorem}
\begin{remark}
The condition  that $\stab_{1-
\delta}(f_i)\leq \delta$ (see Definition~\ref{def:stab}) serves as a convenient proxy for the condition that
the function $f_i$ is essentially of high degree.  Indeed,
if $\stab_{1-
\delta}(f_i)\leq \delta$ then the Fourier mass of $f_i$ on degrees less than $\frac{1}{\delta}$ is at most $O(\delta)$. Conversely, if the Fourier mass
on degrees less than $\frac{1}{\delta} \log(\frac{2}{\delta})$ is at most $\frac{\delta}{2}$, then $\stab_{1-
\delta}(f_i)\leq \delta$.
\end{remark}

\paragraph{Implications to CSP approximability:}  combining~\cref{thm:kcsp} with the transformation from \cite{BKM1} converting integrality gap instances to dictatorship tests, we have the following result:
\begin{theorem}
		\label{thm:dict_test}
		Let $P\colon \Sigma^k\to\{0,1\}$ be a predicate that satisfies the following conditions: (1) $P$ admits no Abelian embedding,  (2a) there exists an instance of Max-$P$-CSP that has a $(1,s)$-integrality gap for the basic SDP relaxation, (2b) on every constraint, the local distribution
in the SDP solution
admits no Abelian embedding. Then, for every $\eps>0$, there is a dictatorship test for $P$-CSP that has perfect completeness and soundness $s+\eps$.
\end{theorem}

In other words,~\cref{thm:dict_test} gives some evidence that for predicates with no Abelian embedding, the best possible approximation algorithm is still the natural SDP relaxation with an appropriate rounding scheme.

\subsubsection{Local Inverse Theorems under Milder Assumptions}
The proof of~\cref{thm:dict_test} requires a structural result for functions $f_1,\ldots,f_k$ 
exhibiting $k$-wise correlation under a milder assumption on $\mu$. Namely, instead of requiring that $\mu$ does not admit an Abelian embedding, we require it to be \emph{pairwise-connected} as well as that its marginal on the first $k-1$ coordinates, denoted by $\mu_{-k}$, admits no Abelian embedding.
\begin{definition}
    A distribution $\mu$ over $\Sigma_1\times\ldots\times\Sigma_k$ is pairwise-connected if for all $i\neq j$ the bipartite graph $(\Sigma_i\cup \Sigma_j, \supp(\mu_{i,j}))$ is connected.
\end{definition}
In this setting we have the following result: 
\begin{restatable}{lemma}{kcsp}
\label{lemma:kcsp}
Let $k$ be a positive integer, let $\mu$ be a pairwise-connected distribution over $\A_1 \times \dots \times \A_k$ in which the probability of each atom is at least $\alpha$, and assume that $\mu_{-k}$ admits no Abelian embeddings. Then, for every $\eps > 0$ there is $\delta := \delta(\alpha,\eps)>0$ such that for all sufficiently large $n$, if $1$-bounded functions $f_i: \A_i^n \to \bbC$ satisfy
\[ \left|\E_{(x_1,\dots,x_k) \sim \mu^{\otimes n}}\Big[\prod_{i=1}^k f_i(x_i)\Big]\right| \ge \eps, \]
then a random restriction of $f_1$ correlates to a product function, i.e.,
\[ \Pr_{I \sim_{1-\delta} [n], z \sim \mu_1^I} \left[\exists \{P_i: \A_1 \to \bbC, \|P_i\|_\infty \le 1\}_{i \in \bar{I}} \enspace \text{ with } \enspace \Big|\E_{x \sim \mu_1^{\bar{I}}}\Big[(f_1)_{I\to z}(x) \prod_{i \in \bar{I}} P_i(x_i) \Big]\Big| \ge \delta \right] \ge \delta. \]
 Quantitatively, $\delta = \exp(-\exp(\dots\exp(\eps^{-O_{\alpha}(1)})))$ where the number of exponentials is $k^{O(k)}$.
\end{restatable}

If we strengthen the assumption on $\mu_{-k}$, and assume that is is connected, then we can improve upon
the quantitative aspect of~\cref{lemma:kcsp}, and get the following result:
\begin{lemma}
\label{lemma:kcsp_conn}
Let $k$ be a positive integer, let $\mu$ be a pairwise-connected distribution over $\A_1 \times \dots \times \A_k$ in which the probability of each atom is at least $\alpha$, and assume that $\mu_{-k}$ is connected. Then, for every $\eps > 0$ there is a constant $\delta := \delta(\alpha,\eps)>0$ such that for all sufficiently large $n$, if $1$-bounded functions $f_i: \A_i^n \to \bbC$ satisfy
\[ \left|\E_{(x_1,\dots,x_k) \sim \mu^{\otimes n}}\Big[\prod_{i=1}^k f_i(x_i)\Big]\right| \ge \eps, \]
then a random restriction of $f_1$ correlates to a product function, i.e.,
\[ \Pr_{I \sim_{1-\delta} [n], z \sim \mu_1^I} \left[\exists \{P_i: \A_1 \to \bbC, \|P_i\|_\infty \le 1\}_{i \in \bar{I}} \enspace \text{ with } \enspace \Big|\E_{x \sim \mu_1^{\bar{I}}}\Big[(f_1)_{I\to z}(x) \prod_{i \in \bar{I}} P_i(x_i) \Big]\Big| \ge \delta \right] \ge \delta. \]
 Quantitatively, we can take $\delta = \exp(-\eps^{-O_{\alpha}(1)})$.
\end{lemma}

\subsubsection{A Global Inverse Theorem}
Finally, we state and prove a global inverse theorem 
in the setting of~\cref{lemma:kcsp}. By a global inverse theorem, we mean a result that 
makes an assertion regarding the global structure 
of the functions $f_i$'s, as opposed to only a local one after random restrictions. Indeed, combining~\cref{lemma:kcsp} and the restriction inverse theorem, namely~\cite[Theorem 9]{csp6}, we get the following conclusion.

\begin{theorem}
\label{thm:intro_kcsp_global}
Let $k$ be a positive integer, let $\mu$ be a distribution over $\A_1 \times \dots \times \A_k$ in which the probability of each atom is at least $\alpha$, and assume that $\mu_{-k}$ admits no Abelian embeddings. Then, for every $\eps > 0$ there are $\delta := \delta(\alpha, \eps)>0$ and $d = d(\alpha,\eps)\in\mathbb{N}$ such that for all sufficiently large $n$, if $1$-bounded functions $f_i: \A_i^n \to \bbC$ satisfy
\[ \left|\E_{(x_1,\dots,x_k) \sim \mu^{\otimes n}}\Big[\prod_{i=1}^k f_i(x_i)\Big]\right| \ge \eps, \]
then there exists $L\colon \Sigma_1^n\to\mathbb{C}$ of degree at most $d$ and $\|L\|_2\leq 1$, as well as 
a product function $P\colon \Sigma_1^n\to\mathbb{C}$ of the form $P(x) = \prod\limits_{i=1}^{n}P_i(x_i)$ where $|P_i(x_i)| = 1$ for all $i$ and $x$, such that 
$|\langle f, L\cdot P\rangle|\geq \delta$. 
\end{theorem}

\subsection{Subsequent and Future Works}
\paragraph{Subsequent works:} in~\cite{BKLMDHJ3}, we use~\cref{lemma:kcsp_conn} with $k=4$ to give the first reasonable bounds for the density Hales-Jewett theorem on $\{0,1,2\}^n$. The density Hales-Jewett problem refers to the problem of determining the maximum density of a subset $S \subseteq \{0,1,\ldots,k-1\}^n$ that avoids all {\em combinatorial lines} of length $k$. Here, we say that a $k$-tuple $x^{(1)}, \dots, x^{(k)} \in S$ forms a combinatorial line of length $k$ if not all $x^{(j)}$ are equal, and for each coordinate $i = 1, 2, \dots, n$, 
\[
(x^{(1)}_i,\ldots,x^{(k)}_i)
\in\left\{(0,\ldots,0),(1,\ldots,1),\ldots,(k-1,\ldots,k-1),(0,1,2,\ldots,k-1)\right\}.
\]
The density Hales-Jewett theorem, proved by Furstenberg and Katznelson \cite{FK89,FK91}, asserts that for every $k\in\mathbb{N}$ and $\delta > 0$, for sufficiently large $n\geq n_0(k,\delta)$, a subset $S \subseteq \{0,1,\dots,k-1\}^n$ of density at least $\delta$ must contain a combinatorial line of length $k$. The proof of Furstenberg and Katznelson is ergodic theoretic and gives no quantitative bounds on $n_0$. The Polymath project~\cite{polymath2012new} provided an elementary proof that does give quantitative bounds. For $k=3$, their result asserts that if $3^{-n}|S| \ge \Omega((\log^* n)^{-1/2})$, then $S$ contains a combinatorial line. 
In~\cite{BKLMDHJ3}, we improve upon this result by showing that if $S\subseteq \{0,1,2\}^n$ is such that $3^{-n}|S| \ge \Omega((\log\log\log\log n)^{-c})$ (for some absolute constant $c>0$), then $S$ contains a combinatorial line. 

\begin{remark}
~\cref{lemma:kcsp} is already good enough to get an effective bounds for the density Hales-Jewett theorem in $[3]^n$, except that the bound would have five logs as opposed to four. 
\end{remark}
\vspace{-1ex}
\paragraph{Future works:} the results of this paper 
can be seen as an extension of the results in~\cite{BKM2} to the setting of $k$-ary distributions with no Abelian embeddings. 
In future works, we plan to investigate the class
of $k$-ary distributions that do admit Abelian embeddings, and prove results along the lines of~\cite{BKM4,csp6} characterizing $k$-tuples of
functions that achieve non-trivial $k$-wise correlation. This 
scenario already includes within it,
(as special cases) the Gowers' uniformity norms over finite fields~\cite{Gowers01,green2008inverse,bergelson2010inverse,tao2012inverse,gowers2017quantitative}, and such inverse theorems may lead to 
further progress in related problems in additive combinatorics.

\section{Preliminaries}
\paragraph{Notations:} 
We denote $[n] = \{1,\ldots,n\}$, and for 
$\alpha\in (0,1)$ we denote by $I\sim_{1-\alpha} [n]$ the distribution over subsets of $[n]$ in which each element $i\in [n]$ is included with probability $1-\alpha$. For a real number $x$ we write $\|x\|_{\R/\Z} := \min_{z \in \Z} |x-z|$.
\label{sec:prelim}
\subsection{Noise Operators}
Let $(\A, \nu)$ be a probability space. Define the inner product on this space by $\langle f, g\rangle_\nu := \E_{x\in \nu}[f(x)\bar{g(x)}]$. We need the following definition of the noise operator.
\begin{definition}
	Let $\A$ be a finite alphabet, and $\nu$ be a measure on $\A$. For a parameter $\rho\in [0,1]$,
	we define the $\rho$-correlated distribution with respect to $\nu$ as follows. For any
	$y\in\A$, the distribution of inputs that are $\rho$-correlated with $y$ is denoted by $y'\sim \mathrm{T}_{\rho} y$
	and is defined by taking $y' = y$ with probability $\rho$, and otherwise sampling $y'\sim \nu$.
\end{definition}

We also view $\mathrm{T}_{\rho}$ as an operator on functions, mapping $L^2(\A,\nu)$ to $L^2(\A,\nu)$
defined as
\[
(\mathrm{T}_{\rho} g)(y) = \E_{y'\sim \mathrm{T}_{\rho} y}\left[ {g(y')}\right].
\]
We then tensorize this operator, i.e., consider $\mathrm{T}_{\rho}^{\otimes n}$ which acts on functions on $n$-variables.
When clear from context, we drop the $\otimes n$ superscript from notation.
The stability of a function $f: \Sigma^n \rightarrow \mathbb{C}$ with respect to the noise parameter $\rho$ is defined as follows.

\begin{definition}\label{def:stab}
	$\mathsf{Stab}^\nu_{\rho}(g) := \langle g, \mathrm{T}_{\rho} g\rangle_{\nu^{\otimes n}}$. We often drop the superscript $\nu$ from $\mathsf{Stab}^\nu_{\rho}(g)$ when it is clear from the context.
\end{definition}

\subsection{Inverse Theorems}
In this section, we state the inverse theorems that we need to prove our results. 
Our proof of~\cref{thm:kcsp} is by induction on $k$, and the base case $k=3$ was proven in~\cite{BKM2}, as follows:
 \begin{theorem} 
\label{thm:3csp_bkm2}
Let $\mu$ be a distribution over $\A_1 \times \A_2 \times \A_3$ that does not admit any Abelian embedding, and in which the probability of each atom is at least $\alpha$. Then, for all $\eps > 0$ there is a constant $\delta := \delta(\alpha,\eps)$ such that for all sufficiently large $n$, if $1$-bounded functions $f_i: \A_i^n \to \bbC$ satisfy
\[ \left|\E_{(x_1, x_2,x_3) \sim \mu^{\otimes n}}\Big[\prod_{i=1}^3 f_i(x_i)\Big]\right| \ge \eps, \]
then $\stab_{1-\delta}(f_i) \ge \delta$ for all $i = 1, 2, 3$. Quantitatively, $\delta(\eps) \ge \exp(-\eps^{-O_{\alpha}(1)})$.
\end{theorem}

We will also need the following local inverse theorem from~\cite{csp6}. Here, the only assumption on the distribution $\mu$ is that it is pairwise-connected. 
\begin{theorem}
\label{thm:3csp}
Let $\mu$ be a pairwise-connected distribution over $\A \times \B \times \C$, and assume that the probability of each atom is at least $\alpha$. If $1$-bounded functions $f: \A^n \to \bbC$, $g: \B^n \to \bbC$, $h: \C^n \to \bbC$ satisfy that
\[ \left|\E_{(x,y,z) \sim \mu^{\otimes n}}[f(x)g(y)h(z)] \right| \ge \eps, \] 
then there is a constant $\delta := \delta(\alpha,\eps) > 0$ and distribution $\nu$ such that $\mu = (1-\delta)\nu + \delta U$, where $U$ is uniform over $\A$, such that:
\[ \Pr_{I \sim_{1-\delta} [n], z \sim \nu^I}\left[\exists \{P_i: \A \to \bbC, \|P_i\|_\infty \le 1\}_{i \in \bar{I}} \enspace \text{ with } \enspace \Big|\E_{x \sim \A^{\bar{I}}}\Big[f_{I\to z}(x) \prod_{i \in \bar{I}} P_i(x_i) \Big]\Big| \ge \delta \right] \ge \delta. \]
Quantitatively, $\delta(\eps) \ge \exp(-\eps^{-O_{\alpha}(1)})$.
\end{theorem}

Finally, in order to get improved bounds in Lemma~\ref{lemma:kcsp_conn}, instead of Theorem~\ref{thm:3csp_bkm2}, we use the following rest of Mossel~\cite{Mossel10}. 

\begin{lemma}[Lemma 6.2 from~\cite{Mossel10}]
\label{lemma:mossel}
Let $k$ be a positive integer and $\mu$ be a distribution over $\A_1 \times \dots \times \A_k$ such that $\mu$ is {\em connected} and the probability of each atom is at least $\alpha$. 
Then, for all $\eps > 0$ there is a constant $\delta := \delta(\alpha,\eps)$ such that for all sufficiently large $n$, if $1$-bounded functions $f_i: \A_i^n \to \bbC$ satisfy
\[ \left|\E_{(x_1,\dots,x_k) \sim \mu^{\otimes n}}\Big[\prod_{i=1}^k f_i(x_i)\Big]\right| \ge \eps, \]
then $\stab_{1-\delta}(f_i) \ge \delta$ for all $i \in [k]$. Quantitatively, $\delta(\eps, \alpha) \ge (\eps\cdot \alpha)^{O(1)}$.
\end{lemma}

\section{$k$-CSPs Without Abelian Embeddings}
\label{sec:kcsp}

In this section we prove \cref{thm:kcsp}. The proof proceeds by induction on $k$, and critically uses \cref{thm:3csp_bkm2} and \cref{thm:3csp} to perform the inductive step. 

\paragraph{The structure of the proof and notations:}
the base case of the proof, namely the case of $k=3$ in~\cref{thm:kcsp}, is given by~\cref{thm:3csp_bkm2}. We
will show that for each $k$,~\cref{thm:kcsp} for
$k$ implies~\cref{lemma:kcsp} for $k+1$, which in turn implies~\cref{thm:kcsp} for $k+1$. Thus,
we denote by $\delta_k(\alpha,\eps)$ the parameter of~\cref{thm:kcsp} for $\eps$, and by 
$\tau_k(\alpha,\eps)$ the parameter $\delta$ from~\cref{lemma:kcsp}
for $\eps$. With this, we have that $\delta_3(\alpha,\eps) = \exp(-\eps^{-O_{\alpha}(1)})$, and we will show bounds for $\tau_k(\alpha,\eps)$ in terms of $\delta_k(\alpha,\eps)$, and bounds for $\delta_{k+1}(\alpha,\eps)$ in terms of $\tau_k(\alpha,\eps)$.
We often suppress the dependency on $\alpha$ to simplify notation.

\subsection{Preliminary Observations}
\label{subsec:obs}

We start with the following two simple observations.
\begin{obs}
\label{obs:pc}
If the support of $\mu$ admits no Abelian embedding, then $\mu$ is pairwise-connected.
\end{obs}
\begin{proof}
Let $\mu_{ij}$ be the restriction of $\mu$ to the $i,j$ coordinates. If $(\Sigma_i\cup\Sigma_j, \supp(\mu_{ij}))$ is not connected, then there is a non-trivial partition $\A_i = \A_i' \cup \A_i''$ and $\A_j = \A_j' \cup \A_j''$ such that $\supp(\mu_{ij}) \subseteq (\A_i' \times \A_j') \cup (\A_i'' \times \A_j'')$. Now let $\sigma_t \equiv 0$ for $t \neq i, j$, $\sigma_i(x) = 1_{x \in \A_i'}$, and $\sigma_j(x) = -1_{x \in \A_j'}$. This is clearly an Abelian embedding, and at least one of $\sigma_i$ or $\sigma_j$ is nonconstant.
\end{proof}

\begin{obs}
\label{obs:sub}
If the support of a distribution $\mu$ on $\A_1 \times \dots \times \A_k$ has no Abelian embedding, then the restriction of $\mu$ to any $k-1$ coordinates also has no Abelian embedding.
\end{obs}
\begin{proof}
Let $\mu'$ be the restriction of $\mu$ to the coordinates $1, \dots, k-1$. If $\mu'$ has Abelian embeddings $\sigma_1', \dots, \sigma_{k-1}'$ then $\mu$ has the Abelian embeddings $\sigma_k \equiv 0$ and $\sigma_i \equiv \sigma_i'$ for $i = 1, \dots, k-1$.
\end{proof}

\subsection{Reducing to a $(k-1)$-ary Instance}
\label{subsec:reduce}



The goal of this subsection and of~\cref{subsec:apply3csp} is to show that~\cref{thm:kcsp} for $k-1$ implies \cref{lemma:kcsp} for $k$. The first step is to apply Cauchy-Schwarz to remove the $x_k$ variable, with the goal of reducing to a $(k-1)$-ary instance. To do this, we need to set up a bit of notation. Recall that $\mu_{-k}$ denoted the marginal distribution of $\mu$ on $\A_1 \times \dots \times \A_{k-1}$. 
Let $\mu_{-k,-k}$ be the following distribution on $(\A_1 \times \dots \times \A_{k-1}) \times (\A_1 \times \dots \times \A_{k-1})$:
\begin{itemize}
    \item Sample $x_k$ from $\mu_k$.
    \item Sample $(x_1, \dots, x_{k-1})$ and $(x_1', \dots, x_{k-1}')$ from $\mu$ conditioned on $x_k$.
    \item $\mu_{-k,-k}$ contains $((x_1, \dots, x_{k-1}), (x_1', \dots, x_{k-1}'))$.
\end{itemize}
Note that $\supp(\mu_{-k,-k})$ contains $((x_1,\dots,x_{k-1}),(x_1,\dots,x_{k-1}))$ for any $(x_1,\dots,x_{k-1}) \in \mu_{-k}$. In other words, $\mu_{-k,-k}$ essentially contains $\mu_{-k}$ as a submeasure.

Now, using the hypothesis of \cref{lemma:kcsp} gives us:
\begin{align*}
    \eps^2 &\le \left|\E_{(x_1,\dots,x_k) \sim \mu^{\otimes n}}\Big[\prod_{i=1}^k f_i(x_i)\Big]\right|^2 \le \E_{\substack{(x_1,\dots,x_{k-1}) \\ (x_1',\dots,x_{k-1}')} \sim \mu_{-k,-k}^{\otimes n}} \Big[\prod_{i=1}^{k-1} f_i(x_i)\bar{f_i(x_i')}\Big].
\end{align*}
By abuse of notation, identify between $\mu_{-k}$  and the measure on $(\A_1 \times \dots \times \A_{k-1}) \times (\A_1 \times \dots \times \A_{k-1})$ consisting of points 
\[ 
((x_1,\dots,x_{k-1}), (x_1, \dots, x_{k-1})), 
\] 
where $(x_1,\dots,x_{k-1})$ is distributed as in $\mu_{-k}$. 
We may write $\mu_{-k,-k} = \alpha^2 \mu_{-k} + (1-\alpha^2) \nu$
where $\nu$ is a distribution. Applying this random restriction to the above equation gives:
\begin{align} \eps^2 \le \E_{I \sim_{1-\alpha^2} [n]} \E_{\substack{(z_1,\dots,z_{k-1}) \\ (z_1', \dots, z_{k-1}')} \sim \nu^{\otimes I}} \E_{(x_1,\dots,x_{k-1}) \sim \mu_{-k}^{\otimes \bar{I}}} \Big[\prod_{i=1}^{k-1} (f_i)_{I \to z_i}(x_i) \bar{(f_i)_{I \to z_i'}(x_i)} \Big]. \label{eq:stuff1}
\end{align}
By the hypothesis, $\mu_{-k}$ has no Abelian embedding, and thus if
\begin{align} \left|\E_{(x_1,\dots,x_{k-1}) \sim \mu_{-k}^{\otimes \bar{I}}} \Big[\prod_{i=1}^{k-1} (f_i)_{I \to z_i}(x_i) \bar{(f_i)_{I \to z_i'}(x_i)} \Big] \right| \ge \eps^2/2, \label{eq:stuff2} 
\end{align}
by induction there is a constant $\delta_{k-1} = \delta_{k-1}(\eps^2/2)$ such that
\[ \stab_{1-\delta_{k-1}}\left((f_1)_{I\to z_1}\bar{(f_1)_{I \to z_1'}} \right) \ge \delta_{k-1}. \]
Combining this with \eqref{eq:stuff1} gives:
\begin{equation}
\E_{I \sim_{1-\alpha^2} [n]} \E_{\substack{(z_1,\dots,z_{k-1}) \\ (z_1', \dots, z_{k-1}')} \sim \nu^{\otimes I}}\left[\stab_{1-\delta_{k-1}}\left((f_1)_{I\to z_1}\bar{(f_1)_{I \to z_1'}} \right) \right] \ge \eps^2\delta_{k-1}/2. \label{eq:f1}
\end{equation}

\subsection{Applying the $3$-ary Inverse Theorem}
\label{subsec:apply3csp}
Let $\A := \A_1$. Now we will interpret \eqref{eq:f1} as a $3$-wise correlation over a pairwise-connected distribution, and then apply \cref{thm:3csp}. Let $\A^+ = (\A \times \A) \cup \{\star\}$, for some symbol $\star$.

Let $\xi$ be the following distribution over $\A \times \A \times \A^+$:
\begin{itemize}
    \item With probability $1-\alpha^2$, sample $(x_1,x_1') \sim \nu_1$, and include $(x_1,x_1',(x_1,x_1'))$ in $\xi$.
    \item Otherwise, let $x \sim \mu_1$. With probability $1-\delta_{k-1}$ include $(x,x,(x,x))$ in $\xi$, and with probability $\delta_{k-1}$ include $(x,x,\star)$ in $\xi$.
\end{itemize}
Finally, given $x^+ \in \A^+$, define the following way of sampling $(x,x') \in \A^n \times \A^n$, which we denote as $(x,x') \sim_{\star} x^+$. If $x^+_i \neq \star$, let $x_i, x_i'$ be the first, second coordinates of $x^+_i$ respectively. Otherwise if $x^+_i = \star$, sample $x \sim \mu_1$ and set $x_i = x_i' = x$. This lets us define the function $g: (\A^+)^n \to \bbC$ as
\[ g(x^+) := \E_{(x,x') \sim_{\star} x^+}\Big[f_1(x)\bar{f_1(x')} \Big]. \]
At this point, let us collect two useful observations:
\begin{obs}
\label{obs:pcxi}
The distribution $\xi$ is pairwise-connected.
\end{obs}
\begin{proof}
Indeed, the $13$-coordinates are pairwise-connected because every $x \in \A$ is connected to $\star \in \A^+$ because $(x,x,\star)$ is in $\supp(\xi)$ for all $x \in \A$. Similarly, the $23$-coordinates are pairwise-connected. For the $12$-coordinates, note that $(x,x')$ is in the support if there is some $x_k \in \A_k$ such that $(x, x_k), (x', x_k) \in \mu_{1k}$. Thus, the claim follows because the $1k$-coordinates were connected in $\mu$ by the hypothesis of \cref{lemma:kcsp}.
\end{proof}

\begin{obs}
\label{obs:3wisexi}
It holds that
\begin{align*}
    \E_{(x,x',x^+) \sim \xi^{\otimes n}}\Big[f_1(x)\bar{f_1(x')}\bar{g(x^+)} \Big] &= \E_{I \sim_{1-\alpha} [n]} \E_{\substack{(z_1,\dots,z_{k-1}) \\ (z_1', \dots, z_{k-1}')} \sim \nu^{\otimes I}}\left[\stab_{1-\delta_{k-1}}\left((f_1)_{I\to z_1}\bar{(f_1)_{I \to z_1'}} \right) \right] \\
    &\ge \eps^2\delta_{k-1}/2.
\end{align*}
\end{obs}
\begin{proof}
    The equality is by the definition of $\xi$ and $g$. The inequality is exactly \eqref{eq:f1}.
\end{proof}
Given these, we can establish \cref{lemma:kcsp}  using~\cref{thm:3csp}.
\begin{proof}[Proof that~\cref{thm:kcsp} with $k-1$ implies~\cref{lemma:kcsp} with $k$]
Note that all atoms in $\xi$ have mass at least $\alpha^2\delta_{k-1}$, and that $\xi$ is pairwise-connected by \cref{obs:pcxi}. Thus,
we may write $\xi = \zeta \xi' + (1-\zeta)\xi''$
where $\xi',\xi''$ are distributions, $\zeta\geq \alpha^2 \delta_{k-1}$, the probability 
of each atom in $\xi'$ is at least 
$\alpha^2$ and $\xi'$ has the same support as $\xi$. Applying 
random restrictions, we get that with probability at least $\eps^2\delta_{k-1}/4$, choosing $J\sim_{1-\zeta} I$ and $(a,b,c)\sim {\xi''}^{J}$ we have that 
\[
\left| 
\E_{(x,x',x^+) \sim {\xi'}^{\overline{J}}}\Big[
(f_1)_{J\rightarrow a} (x)
\bar{(f_1)_{J\rightarrow b}(x')}
\bar{(g)_{J\rightarrow c}(x^+)} \Big]
\right|\geq \eps^2\delta_{k-1}/4.
\]

For each such $J$ and $(a,b,c)$,~\cref{thm:3csp} gives that:
\[ 
\Pr_{I \sim_{1-\delta} J, z \sim \nu_1^I} \left[\exists \{P_i: \A_1 \to \bbC, \|P_i\|_\infty \le 1\}_{i \in \bar{I}} \enspace \text{ with } \enspace \Big|\E_{x \sim \A_1^{\bar{I}}}\Big[(f_1)_{J\rightarrow a, I\to z}(x) \prod_{i \in \bar{I}} P_i(x_i) \Big]\Big| \ge \delta \right] \ge \delta, 
\]
where $\delta \ge \exp(-(\eps\delta_{k-1})^{-O_{\alpha}(1)})$. 
Combining the two random restrictions, we see that
we may write $\mu_1 = (1-\delta\zeta)\nu + \delta\zeta U$ and get that
\[ 
\Pr_{I \sim_{1-\delta\zeta} [n], z \sim \nu_1^I} \left[\exists \{P_i: \A_1 \to \bbC, \|P_i\|_\infty \le 1\}_{i \in \bar{I}} \enspace \text{ with } \enspace \Big|\E_{x \sim \A_1^{\bar{I}}}\Big[(f_1)_{I\to z}(x) \prod_{i \in \bar{I}} P_i(x_i) \Big]\Big| \ge \delta \right] \ge \delta, 
\] 
By applying a further random restriction we can convert $\nu$ and $U$ back into $\mu_1$ at the cost of decreasing $\delta$ by a constant. This finishes the inductive step, establishing that
\begin{equation}\label{eq:recursion1}
\tau_{k}(\eps)
\geq 
\Omega(\eps^2\delta_{k-1}(\eps^2/2))
\cdot
\exp(-(\eps\delta_{k-1}(\eps^2/2))^{-O_{\alpha}(1)})
\geq 
\exp(-\delta_{k-1}(\eps^2/2)^{-O_{\alpha}(1)}).
\end{equation}
In~\cref{subsec:conclude} we will establish a recursive relation for $\delta_{k-1}$, and then in~\cref{sec:solve_rec} we conclude a quantitative estimate for $\tau_k$.
\end{proof}

We get~\cref{lemma:kcsp_conn} using a similar argument.
\begin{proof}[Proof of \cref{lemma:kcsp_conn}]
The proof of this lemma is along the lines of the proof of the Lemma~\ref{lemma:kcsp}. The only change is that, instead of relying on the inductive hypothesis to go from \eqref{eq:stuff1} to \eqref{eq:f1}, we use Lemma~\ref{lemma:mossel} where $\delta_{k-1} \ge  \eps^{O_{\alpha}(1)}$. The conclusion then follows from Observation~\ref{obs:3wisexi} and applying Theorem~\ref{thm:3csp}.
\end{proof}

\subsection{Proof of~\cref{thm:kcsp}}
\label{subsec:conclude}
Fix functions $f_1,\ldots,f_k$ as in the statement of~\cref{thm:kcsp}.
The goal of this section is to replace each $f_i$ by a product function (potentially after random restriction) by applying \cref{lemma:kcsp}. Then, by using the fact that $\mu$ has no Abelian embedding, we will conclude that the product functions must be nearly constant, which implies that $f_i$ correlates to a constant after random restriction. This is essentially what \cref{thm:kcsp} claims. 

More precisely, the following lemma allows us to replace $f_k$ by a product function:
\begin{lemma}
\label{lemma:replace}
Let $k$ be a positive integer and let $\mu$ be a distribution over $\A_1 \times \dots \times \A_k$ 
that has no Abelian embeddings and in which the probability of each atom is at least $\alpha$. 
Then, for all $\eps > 0$, if $1$-bounded functions functions $f_i: \A_i^n \to \bbC$ satisfy
\[ \left|\E_{(x_1,\dots,x_k) \sim \mu^{\otimes n}}\Big[\prod_{i=1}^k f_i(x_i)\Big]\right| \ge \eps, \]
then with probability at least $\tau_k/2 = \tau_k(\alpha,\eps^2)/2$ over $I \sim_{1-\tau_k} [n]$ and $(z_1, \dots, z_k) \sim \mu^I$, there is a $1$-bounded product function $P_k: \A_k^{\bar{I}} \to \bbC$ satisfying
\[ \left|\E_{x \sim \mu^{\otimes \bar{I}}}\Big[\prod_{i=1}^{k-1} (f_i)_{I \to z_i}(x_i) P_k(x_k) \Big] \right| \ge 
\frac{\tau_k}{2}.
 \]
\end{lemma}
\begin{proof}
Define the function $\wt{f}_k: \A_k^n \to \bbC$ as
\[ \wt{f}_k(x) := \E_{(x_1,\dots,x_k) \sim \mu^{\otimes n}} \Big[\prod_{i=1}^{k-1} f_i(x_i) ~\Big |~ x_k = x\Big]. \]
Note that $\wt{f}_k$ is $1$-bounded, and by the Cauchy-Schwarz inequality
\begin{align*}
    \eps^2 &\le \left|\E_{(x_1,\dots,x_k) \sim \mu^{\otimes n}}\Big[\prod_{i=1}^k f_i(x_i)\Big]\right|^2 = \left|\E_{x \sim \mu_k^{\otimes n}}\Big[\wt{f}_k(x)f_k(x)\Big]\right|^2 \le \|\wt{f}_k\|_2^2\|f_k\|_2^2 \\
    &\le \|\wt{f}_k\|_2^2 = \left|\E_{(x_1,\dots,x_k) \sim \mu^{\otimes n}}\Big[\prod_{i=1}^{k-1} f_i(x_i) \bar{\wt{f}_k(x_k)}\Big]\right|.
\end{align*}
Note that by \cref{obs:pc} and \cref{obs:sub}, the hypotheses of \cref{lemma:kcsp} hold. Thus, \cref{lemma:kcsp} gives that for some $\tau_k := \tau_k(\alpha,\eps^2)$, with probability at least $\tau_k$ over $I \sim_{1-\tau_k} [n]$ and $z \sim \mu_k^{\otimes I}$ that there is a $1$-bounded product function $P_k: \A_k^{\bar{I}} \to \bbC$ with
\[ \tau_k \le \left|\E_{x \sim \mu_k^{\otimes \bar{I}}}\Big[\bar{(\wt{f}_k)_{I \to z}(x)}P_k(x)\Big] \right| = \E_{(z_1,\dots,z_k) \sim \mu^{\otimes I} : z_k = z} \E_{(x_1,\dots,x_k) \sim \mu^{\otimes \bar{I}}} \Big[\prod_{i=1}^{k-1} (f_i)_{I \to z_i}(x_i) \bar{P_k(x_k)} \Big].  \]
This implies the desired conclusion by an averaging argument.
\end{proof}
Applying \cref{lemma:replace} a total of $k-1$ times gives the following.
\begin{corollary}
\label{cor:replace}
Let $k$ be a positive integer and let $\mu$ be a distribution over $\A_1 \times \dots \times \A_k$ 
that has no Abelian embeddings and in which the probability of each atom is at least $\alpha$. 
Then, for all $\eps > 0$ there is a constant $\delta := \delta(\alpha,\eps)$ such that for all sufficiently large $n$, if $1$-bounded functions functions $f_i: \A_i^n \to \bbC$ satisfy
\[ \left|\E_{(x_1,\dots,x_k) \sim \mu^{\otimes n}}\Big[\prod_{i=1}^k f_i(x_i)\Big]\right| \ge \eps, \]
then with probability at least $\delta$ over $I \sim_{1-\delta} [n]$ and $z \sim \mu_1^I$, there are $1$-bounded product functions $P_i: \A_i^{\bar{I}} \to \bbC$ for $i = 2, \dots, k$ such that:
\[ \left|\E_{x \sim \mu_1^{\otimes n}} \Big[(f_1)_{I \to z}(x) \prod_{i=2}^k P_i(x_i)\Big] \right| \ge \delta. \]
Quantitatively, we can take 
$\delta := \tau_k(\dots(\tau_k(\eps^{2}))^2)$ where the number of applications of $\tau_k$ is $k-1$.
\end{corollary}
We want to use this to argue that $(f_1)_{I \to z}$ has large noise stability. Towards this, we need a somewhat technical statement about when product functions can correlate under distributions with no Abelian embedding. We start with the one-dimensional case.
\begin{lemma}
\label{lemma:prod1}
Let $k$ be a positive integer and let $\mu$ be a distribution over $\A_1 \times \dots \times \A_k$ 
that has no Abelian embeddings and in which the probability of each atom is at least $\alpha$.
For all $\eta > 0$ and $1$-bounded functions $P_i: \A_i \to \bbC$ satisfying
\[ \left|\E_{(x_1,\dots,x_k) \sim \mu} \Big[\prod_{i=1}^k P_i(x_i) \Big] \right| \ge 1-\eta, \]
it holds that $|P_1(x)| \ge 1-O_{\alpha}(\eta)$ and $|P_1(x) - P_1(y)| \le O_{\alpha}(\sqrt{\eta})$ for all $x, y \in \A_1$.
\end{lemma}
\begin{proof}
We think of $\alpha$ as a constant throughout.
We start by reinterpreting what it means for $\mu$ to have no Abelian embedding linear-algebraically. If $\mu$ has no Abelian embedding, then there are no functions $\a_i: \A_i \to \R$ such that $\sum_{i=1}^k \a_i(x_i) \in \Z$ for all $(x_1, \dots, x_k) \in \supp(\mu)$, and $\pi \circ \a_i$ are nonconstant, where $\pi: \R \to \R/\Z$ is the projection map modulo $1$.
Let $(x_1^*, \dots, x_k^*) \in \supp(\mu)$ be an arbitrary point. We may assume that $\a_i(x_i^*) = 0$ for all $i = 1, \dots, k$ by shifting. Let $S = \sum_{i=1}^k |\A_i|-1$, and denote $\a \in \R^S$ as the vector with entries $\a_i(x_i)$ for $i = 1, \dots, k$ and $x_i \neq x_i^*$. For a point $x = (x_1, \dots, x_k) \in \supp(\mu)$, let $a_x \in \Z^S$ be the vector such that $\langle a_x, \a \rangle = \sum_{i=1}^k \a_i(x_i)$. Then, rephrasing the above condition, we get that $\mu$ has no Abelian embedding if and only if $\langle a_x, \a \rangle \in \Z$ for all $x \in \supp(\mu)$ implies that $\a \in \Z^S$. Consider the lattice $L = \mathsf{span}_{\Z}(a_x)$ and let $a^{(1)}, \dots, a^{(S)} \in L$ be a spanning set for $L$. Let $A \in \Z^{S \times S}$ be the matrix with rows $a^{(1)}, \dots, a^{(S)}$. Then, we know that $\mu$ has no Abelian embedding if and only if $A\a \in \Z^S$ implies that $\a \in \Z^S$. Thus, $A^{-1}$ must have integer entries.

Now, we proceed to the main argument. We may assume that $\eta$ is sufficiently small, because obviously $|P_1(x) - P_1(y)| \le 2$ for all $x, y \in \A_1$. We may assume that
\[ \E_{(x_1,\dots,x_k) \sim \mu} \Big[\prod_{i=1}^k P_i(x_i) \Big] \]
is a real number, and thus the real part of $\prod_{i=1}^k P_i(x_i)$ is at least $1 - O(\eta)$ for all $(x_1, \dots, x_k) \in \supp(\mu)$.
Because $P_j$ are $1$-bounded, this implies that $|P_j(x_j)| \ge 1-O(\eta)$ for all $x_j \in \A_j$. If $\a_j(x_j)$ are such that $P_j(x_j)/|P_i(x_j)| = \exp(2\pi i \a_j(x_j))$, we further know that
\[ \Big\|\sum_{j=1}^k \a_j(x_j) \Big\|_{\R/\Z} \le O(\sqrt{\eta}) \]
for all $(x_1, \dots, x_k) \in \supp(\mu)$. For $a_x$ defined as above, there are integers $k_x$ such that $|\langle a_x, \a \rangle - k_x| \le O(\sqrt{\eta})$ for $x \in \supp(\mu)$. Because $a^{(1)}, \dots, a^{(S)}$ are finite integer combinations of $a_x$, we know that there are integers $k^{(i)}$ such that $|\langle a^{(i)}, \a \rangle - k^{(i)}| \le O(\sqrt{\eta})$ for $i = 1, \dots, S$. Because $A^{-1}$ has integral entries, we conclude that $\|\a\|_{\R/\Z} \le O(\sqrt{\eta})$.

This implies the desired conclusion, because
\begin{align*} |P_1(x) - P_1(y)| &\le O(\eta) + |P_1(x)/|P_1(x)| - P_1(y)/|P_1(y)|| \\ &= O(\eta) + |\exp(2\pi i \a_1(x)) - \exp(2\pi i \a_1(y))| \le O(\sqrt{\eta}),
\end{align*}
because $\|\a_1(x) - \a_1(y)\|_{\R/\Z} \le O(\sqrt{\eta})$.
\end{proof}
Now we state the $n$-dimensional case that we need.
\begin{lemma}
\label{lemma:prodn}
Let $k$ be a positive integer and $\mu$ be a distribution over $\A_1 \times \dots \times \A_k$ whose support has no Abelian embedding and in which the probability of each atom is at least $\alpha$. If $1$-bounded product functions $P_i: \A_i^n \to \bbC$ satisfy
\[ \left|\E_{(x_1,\dots,x_k) \sim \mu^{\otimes n}} \Big[\prod_{i=1}^k P_i(x_i) \Big] \right| \ge \delta, \]
then for any $\gamma \in (0, 1)$ it holds that
\[ \E_{\substack{x \sim \mu_1^{\otimes n} \\ y \sim_{1-\gamma} x}}[|P_1(x) - P_1(y)|^2] \le O_{\alpha}(\gamma\log(1/\delta)). \]
\end{lemma}
\begin{proof}
We think of $\alpha$ as a constant throughout.
Let $P_i(x) = \prod_{j=1}^n P_i^{(j)}(x_j)$ for $1$-bounded functions $P_i^{(j)}: \A_i \to \bbC$. Let $\eta_j$ be such that
\begin{equation} \left|\E_{(x_1,\dots,x_k) \sim \mu}\Big[\prod_{i=1}^k P_i^{(j)}(x_i) \Big] \right| = 1 - \eta_j. \label{eq:jcor} \end{equation}
Thus, $\prod_{j=1}^n (1-\eta_j) = \delta$ and hence $\sum_{j=1}^n \eta_j \le O(\log(1/\delta))$.

The following equality holds for all $x, y \in \A_1^n$:
\begin{align*}
&\Big|\prod_{j=1}^n P_1^{(j)}(x_j) - \prod_{j=1}^n P_1^{(j)}(y_j) \Big|^2
\\ = ~&\prod_{j=1}^n |P_1^{(j)}(x_j)|^2 + \prod_{j=1}^n |P_1^{(j)}(y_j)|^2 - \prod_{j=1}^n P_1^{(j)}(x_j)\bar{P_1^{(j)}(y_j)} - \prod_{j=1}^n \bar{P_1^{(j)}(x_j)}P_1^{(j)}(y_j).
\end{align*}
Note that for
\[ \E_{x \sim \mu_1, y \sim_{1-\gamma} x} P_1^{(j)}(x)\bar{P_1^{(j)}(y)} \] is a real number, and thus equals
\begin{align}
\E_{x \sim \mu_1, y \sim_{1-\gamma} x} P_1^{(j)}(x)\bar{P_1^{(j)}(y)} &= \|P_1^{(j)}\|_2^2 - \frac12 \E_{x \sim \mu_1, y \sim_{1-\gamma} x} |P_1^{(j)}(x) - P_1^{(j)}(y)|^2 \notag\\
&\ge \|P_1^{(j)}\|_2^2 - O(\gamma\eta_j) \notag\\
&\ge \|P_1^{(j)}\|_2^2(1 - O(\gamma\eta_j)), \label{eq:pjbound}
\end{align}
where we used~\cref{lemma:prod1}. Thus
\begin{align*}
\E_{\substack{x \sim \mu_1^{\otimes n} \\ y \sim_{1-\gamma} x}}[|P_1(x) - P_1(y)|^2] &= 2 \prod_{j=1}^n \|P_1^{(j)}\|_2^2 - 2 \prod_{j=1}^n \E_{\substack{x \sim \mu_1 \\ y \sim_{1-\gamma} x}}\Big[P_1^{(j)}(x)\bar{P_1^{(j)}(y)} \Big] \\
&\le 2 \prod_{j=1}^n \|P_1^{(j)}\|_2^2\Big(1 - \prod_{j=1}^n (1 - O(\gamma\eta_j)) \Big)\\ 
&\le O(\gamma\log(1/\delta)),
\end{align*}
where we have used \eqref{eq:pjbound}.
\end{proof}

With~\cref{lemma:prodn} in hand, we now show that 
a function satisfying the conclusion of~\cref{cor:replace} must be noise stable.
\begin{lemma}
\label{lemma:fstab}
Let $k$ be a positive integer and let $\mu$ be a distribution over $\A_1 \times \dots \times \A_k$ 
that has no Abelian embeddings and in which the probability of each atom is at least $\alpha$. 
If $1$-bounded function $f: \A_1^n \to \bbC$ and product functions $P_i: \A_i^n \to \bbC$ for $i = 2, \dots, k$ satisfy
\[ \left|\E_{(x_1,\dots,x_k) \sim \mu^{\otimes n}}\Big[f(x_1) \prod_{i=2}^k P_i(x_i) \Big]\right| \ge \delta \] then for $\gamma \le c\delta^2/\log(1/\delta)$ for sufficiently small $c = c(\alpha)>0$, it holds that $\stab_{1-\gamma}(f) \ge \delta^2/4$.
\end{lemma}
\begin{proof}
Define the product function $\wt{P}: \A_1^n \to \bbC$ as
\[ \wt{P}(x) := \E_{(x_1,\dots,x_k) \sim \mu^{\otimes n} : x_1 = x} \Big[\prod_{i=2}^k P_i(x_i) \Big].  \]
Then the following inequalities hold:
\begin{equation}
\delta \le \left|\E_{(x_1,\dots,x_k) \sim \mu^{\otimes n}}\Big[f(x_1) \prod_{i=2}^k P_i(x_i) \Big]\right| = \left|\E_{x \sim \mu_1^{\otimes n}} \Big[f(x)\wt{P}(x) \Big] \right|, \label{eq:wtp}
\end{equation}
and
\[ \left|\E_{(x_1,\dots,x_k) \sim \mu^{\otimes n}}\Big[\wt{P}(x_1) \prod_{i=2}^k P_i(x_i) \Big]\right| \ge \delta^2. \] The latter follows by Cauchy-Schwarz (see the proof of \cref{lemma:replace}). By \cref{lemma:prodn}, we know that for $\gamma \le c\delta^2/\log(1/\delta)$
\begin{align*} 
\left|\E_{x \sim \mu_1^{\otimes n}, y \sim_{1-\gamma} x}\Big[f(x)(\wt{P}(x) - \wt{P}(y)) \Big] \right| &\le \E_{x \sim \mu_1^{\otimes n}, y \sim_{1-\gamma} x}\Big[|\wt{P}(x) - \wt{P}(y)| \Big] \\
&\le \E_{x \sim \mu_1^{\otimes n}, y \sim_{1-\gamma} x}\Big[|\wt{P}(x) - \wt{P}(y)|^2 \Big]^{1/2} \le O(\sqrt{\gamma\log(1/\delta)}) \le \delta/2.
\end{align*}
Combining this with \eqref{eq:wtp} gives
\begin{align*}
\delta/2 &\le \left|\E_{x \sim \mu_1^{\otimes n}, y \sim_{1-\gamma} x} \Big[f(x)\wt{P}(y) \Big] \right| = \left|\E_{x \sim \mu_1^{\otimes n}, y \sim_{1-\gamma} x} \Big[f(y)\wt{P}(x) \Big] \right| \\
&\le \E_{x \sim \mu_1^{\otimes n}} \left|\E_{y \sim_{1-\gamma} x}[f(y)] \right|,
\end{align*}
where the equality uses symmetry of $x, y$ and the inequality is because $\wt{P}$ is $1$-bounded.
One final application of Cauchy-Schwarz gives that $\stab_{1-\gamma}(f) \ge \delta^2/4$ as desired.
\end{proof}
\cref{thm:kcsp} now readily follows by combining \cref{cor:replace} and \cref{lemma:fstab}.
\begin{proof}[Proof of \cref{thm:kcsp}]
Applying \cref{cor:replace} and \cref{lemma:fstab} gives that for some $\delta$ depending on $\eps$,
\[ \stab_{1-\gamma\delta}(f_1) = \E_{I \sim_{1-\delta} [n], z \sim \mu_1^I}\Big[\stab_{1-\gamma}((f_1)_{I \to z}) \Big] \ge \delta^3/4. \]
for $\gamma = c\delta^2/\log(1/\delta)$. This completes the induction, showing that
\begin{equation}\label{eq:recursion2}
\delta_{k}(\eps)
\geq c\cdot \tau_{k}(\tau_k(\ldots(\tau_k(\eps^2)^2))^2)^4,
\end{equation}
where the number of compositions of $\tau_k$ is 
$k-1$.
\end{proof}

\subsection{The Quantitative Bounds for~\cref{lemma:kcsp,thm:kcsp}}\label{sec:solve_rec}
Let $H(k)$ be the number of exponentials needed in~\cref{thm:kcsp} with $k$.
~\cref{thm:3csp_bkm2} gives us that $H(3)\leq 1$. 
For any $k\geq 3$,~\eqref{eq:recursion1} implies that~\cref{lemma:kcsp} for $k$ holds with $H(k-1)+1$
exponentials, and then~\eqref{eq:recursion1} implies
\cref{thm:kcsp} for $k$ holds with $(k-1)(H(k-1)+1)$
exponentials. We conclude that $H(k)\leq (k-1)(H(k-1)+1)$, and solving gives
\[
H(k)\leq (k-1)+(k-1)(k-2)+(k-1)(k-2)(k-3)+\ldots+(k-1)(k-2)\cdots 3
\leq k^{k}.
\]



{\small
\bibliographystyle{alpha}
\bibliography{refs}}

\newcommand{\etalchar}[1]{$^{#1}$}
\begin{thebibliography}{BKLM24b}

\bibitem[ALM{\etalchar{+}}98]{AroraLMSS1998}
Sanjeev Arora, Carsten Lund, Rajeev Motwani, Madhu Sudan, and Mario Szegedy.
\newblock Proof verification and the hardness of approximation problems.
\newblock {\em Journal of the ACM}, 45(3):501--555, 1998.
\newblock (Preliminary version in {\em 33rd FOCS}, 1992).

\bibitem[AS98]{AroraS1998}
Sanjeev Arora and Shmuel Safra.
\newblock Probabilistic checking of proofs: A new characterization of~{NP}.
\newblock {\em Journal of the ACM}, 45(1):70--122, January 1998.
\newblock (Preliminary version in {\em 33rd FOCS}, 1992).

\bibitem[BK21]{BhangaleK21}
Amey Bhangale and Subhash Khot.
\newblock {Optimal Inapproximability of Satisfiable k-LIN over Non-Abelian
  Groups}.
\newblock In {\em Proceedings of the $53$rd Annual ACM SIGACT Symposium on
  Theory of Computing (STOC)}, page 1615–1628, 2021.

\bibitem[BKLM22]{BravermanKLM21}
Mark Braverman, Subhash Khot, Noam Lifshitz, and Dor Minzer.
\newblock An invariance principle for the multi-slice, with applications.
\newblock In {\em FOCS 2022}, pages 228--236, 2022.

\bibitem[BKLM24a]{csp6}
Amey Bhangale, Subhash Khot, {Yang P.} Liu, and Dor Minzer.
\newblock On approximability of satisfiable $k$-csps: {VI}.
\newblock 2024+.

\bibitem[BKLM24b]{BKLMDHJ3}
Amey Bhangale, Subhash Khot, Yang~P. Liu, and Dor Minzer.
\newblock Reasonable bounds for combinatorial lines of length three.
\newblock 2024+.

\bibitem[BKM21]{BravermanKM21}
Mark Braverman, Subhash Khot, and Dor Minzer.
\newblock {On Rich 2-to-1 Games}.
\newblock In {\em $12$th Innovations in Theoretical Computer Science Conference
  (ITCS)}, volume 185, pages 27:1--27:20, 2021.

\bibitem[BKM22]{BKM1}
Amey Bhangale, Subhash Khot, and Dor Minzer.
\newblock On approximability of satisfiable \emph{k}-csps: {I}.
\newblock In {\em STOC 2022}, pages 976--988. {ACM}, 2022.

\bibitem[BKM23a]{BKM2}
Amey Bhangale, Subhash Khot, and Dor Minzer.
\newblock On approximability of satisfiable $k$-csps: {II}.
\newblock In {\em {STOC} 2023}, pages 632--642. {ACM}, 2023.

\bibitem[BKM23b]{BKM3}
Amey Bhangale, Subhash Khot, and Dor Minzer.
\newblock On approximability of satisfiable k-csps: {III}.
\newblock In {\em {STOC} 2023}, pages 643--655. {ACM}, 2023.

\bibitem[BKM24a]{BKM4}
Amey Bhangale, Subhash Khot, and Dor Minzer.
\newblock On approximability of satisfiable k-csps: {IV}.
\newblock In {\em {STOC} 2024}, pages 1423--1434. {ACM}, 2024.

\bibitem[BKM24b]{BKM5}
Amey Bhangale, Subhash Khot, and Dor Minzer.
\newblock On approximability of satisfiable k-csps: {V}.
\newblock {\em CoRR}, abs/2408.15377, 2024.

\bibitem[BTZ10]{bergelson2010inverse}
Vitaly Bergelson, Terence Tao, and Tamar Ziegler.
\newblock An inverse theorem for the uniformity seminorms associated with the
  action of.
\newblock {\em Geometric and Functional Analysis}, 19(6):1539--1596, 2010.

\bibitem[Bul17]{Bulatov17}
Andrei~A. Bulatov.
\newblock {A Dichotomy Theorem for Nonuniform CSPs}.
\newblock In {\em IEEE $58$th Annual Symposium on Foundations of Computer
  Science (FOCS)}, pages 319--330, 2017.

\bibitem[FGL{\etalchar{+}}96]{FGLSS96}
Uriel Feige, Shafi Goldwasser, Laszlo Lov{\'a}sz, Shmuel Safra, and Mario
  Szegedy.
\newblock {Interactive proofs and the hardness of approximating cliques}.
\newblock {\em Journal of the ACM}, 43(2):268--292, 1996.

\bibitem[FK89]{FK89}
H.~Furstenberg and Y.~Katznelson.
\newblock A density version of the {H}ales-{J}ewett theorem for {$k=3$}.
\newblock {\em Discrete Math.}, 75(1-3):227--241, 1989.
\newblock Graph theory and combinatorics (Cambridge, 1988).

\bibitem[FK91]{FK91}
H.~Furstenberg and Y.~Katznelson.
\newblock A density version of the {H}ales-{J}ewett theorem.
\newblock {\em J. Anal. Math.}, 57:64--119, 1991.

\bibitem[FV98]{FederV98}
Tomás Feder and Moshe~Y. Vardi.
\newblock {The Computational Structure of Monotone Monadic {SNP} and Constraint
  Satisfaction: A Study through Datalog and Group Theory}.
\newblock {\em SIAM Journal on Computing}, 28(1):57--104, 1998.

\bibitem[GM17]{gowers2017quantitative}
WT~Gowers and Luka Mili{\'c}evi{\'c}.
\newblock A quantitative inverse theorem for the ${U}^4$ norm over finite
  fields.
\newblock {\em arXiv preprint arXiv:1712.00241}, 2017.

\bibitem[Gow01]{Gowers01}
W.~T. Gowers.
\newblock A new proof of {S}zemer\'edi's theorem.
\newblock {\em Geom. Funct. Anal.}, 11(3):465--588, 2001.

\bibitem[GT08]{green2008inverse}
Ben Green and Terence Tao.
\newblock An inverse theorem for the {G}owers norm.
\newblock {\em Proceedings of the Edinburgh Mathematical Society},
  51(1):73--153, 2008.

\bibitem[GW95]{GW95}
Michel~X. Goemans and David~P. Williamson.
\newblock {Improved Approximation Algorithms for Maximum Cut and Satisfiability
  Problems Using Semidefinite Programming}.
\newblock {\em Journal of the ACM}, 42(6):1115–1145, 1995.

\bibitem[H{\aa}s01]{Hastad2001}
Johan H{\aa}stad.
\newblock Some optimal inapproximability results.
\newblock {\em Journal of the ACM}, 48(4):798--859, July 2001.
\newblock (Preliminary version in {\em 29th STOC}, 1997).

\bibitem[Kho02]{Khot02}
Subhash Khot.
\newblock On the power of unique 2-prover 1-round games.
\newblock In {\em Proceedings of the $34$th Annual ACM symposium on Theory of
  computing (STOC)}, pages 767--775. ACM, 2002.

\bibitem[KKMO07]{KKMO07}
Subhash Khot, Guy Kindler, Elchanan Mossel, and Ryan O’Donnell.
\newblock {Optimal Inapproximability Results for {MAX‐CUT} and Other
  2‐Variable {CSP}s?}
\newblock {\em SIAM Journal on Computing}, 37(1):319--357, 2007.

\bibitem[MOO05]{MOO05}
Elchanan Mossel, Ryan O'Donnell, and Krzysztof Oleszkiewicz.
\newblock Noise stability of functions with low influences: invariance and
  optimality.
\newblock In {\em Proceedings $46$th Annual IEEE Symposium on Foundations of
  Computer Science (FOCS)}, pages 21--30, 2005.

\bibitem[Mos10]{Mossel10}
Elchanan Mossel.
\newblock Gaussian bounds for noise correlation of functions.
\newblock {\em Geometric and Functional Analysis}, 19(6):1713--1756, 2010.

\bibitem[Pol12]{polymath2012new}
DHJ Polymath.
\newblock A new proof of the density {H}ales-{J}ewett theorem.
\newblock {\em Annals of Mathematics}, pages 1283--1327, 2012.

\bibitem[Rag08]{Rag08}
Prasad Raghavendra.
\newblock {Optimal algorithms and inapproximability results for every {CSP}?}
\newblock In {\em Proceedings of the $14$th Annual ACM symposium on Theory of
  computing (STOC)}, pages 245--254, 2008.

\bibitem[Sch78]{Schaefer78}
Thomas~J. Schaefer.
\newblock The complexity of satisfiability problems.
\newblock In {\em Proceedings of the $10$th Annual ACM Symposium on Theory of
  Computing (STOC)}, page 216–226, 1978.

\bibitem[TZ12]{tao2012inverse}
Terence Tao and Tamar Ziegler.
\newblock The inverse conjecture for the {G}owers norm over finite fields in
  low characteristic.
\newblock {\em Annals of Combinatorics}, 16(1):121--188, 2012.

\bibitem[Zhu20]{Zhuk20}
Dmitriy Zhuk.
\newblock {A Proof of the CSP Dichotomy Conjecture}.
\newblock {\em Journal of the ACM}, 67(5), August 2020.

\end{thebibliography}

\end{document}